\documentclass[sigconf]{acmart}

\AtBeginDocument{%
  \providecommand\BibTeX{{%
    \normalfont B\kern-0.5em{\scshape i\kern-0.25em b}\kern-0.8em\TeX}}}
\newcommand\numberthis{\addtocounter{equation}{1}\tag{\theequation}}

\setcopyright{acmcopyright}
\copyrightyear{2020}
\acmYear{2020}
\usepackage{algorithm}
\usepackage{algorithmic}
\usepackage{subcaption}
\usepackage{graphicx}
\usepackage{nicefrac}
\newtheorem{theorem}{Theorem}
\newtheorem{lemma}{Lemma}
\newtheorem{assumption}{Assumption}[section]
\acmConference[DLP-KDD '20]{DLP-KDD '20: 2nd Workshop on Deep Learning Practice for High-Dimensional Sparse Data with KDD 2020}{Aug 24, 2020}{San Degio, CA}
\acmBooktitle{ 2nd Workshop on Deep Learning Practice for High-Dimensional Sparse Data with KDD 2020,
 Aug 24, 2020, San Degio, CA}
\acmPrice{15.00}
\acmISBN{978-1-4503-XXXX-X/18/06}

\begin{document}

\title{DCAF: A Dynamic Computation Allocation Framework for Online Serving System}


\author{Biye Jiang\*, Pengye Zhang\*, Rihan Chen}
\authornote{These three authors contributed equally, corresponding author: Guorui Zhou <guorui.xgr@alibaba-inc.com>}
\author{Binding Dai, Xinchen Luo, Yin Yang, Guan Wang, Guorui Zhou, Xiaoqiang Zhu, Kun Gai}
\affiliation{\institution{Alibaba Group}}

\renewcommand{\shortauthors}{Jiang, Zhang and Chen, et al.}

\begin{abstract}

Modern large-scale systems such as recommender system and online advertising system are built upon computation-intensive infrastructure. The typical objective in these applications is to maximize the total revenue, e.g. GMV~(Gross Merchandise Volume), under a limited computation resource. Usually, the online serving system follows a multi-stage cascade architecture, which consists of several stages including retrieval, pre-ranking, ranking, etc. These stages usually allocate resource manually with specific computing power budgets, which requires the serving configuration to adapt accordingly. As a result, the existing system easily falls into suboptimal solutions with respect to maximizing the total revenue. The limitation is due to the face that, although the value of traffic requests vary greatly, online serving system still spends equal computing power among them. 

In this paper, we introduce a novel idea that online serving system could treat each traffic request differently and allocate "personalized" computation resource based on its value. We formulate this resource allocation problem as a knapsack problem and propose a Dynamic Computation Allocation Framework~(DCAF). Under some general assumptions, DCAF can theoretically guarantee that the system can maximize the total revenue within given computation budget. DCAF brings significant improvement and has been deployed in the display advertising system of Taobao for serving the main traffic. With DCAF, we 
are able to maintain the same business performance with 20\% computation resource reduction.

\end{abstract}




\maketitle

\section{Introduction}

Modern large-scale systems such as recommender system and online advertising are built upon computation-intensive infrastructure \cite{cheng2016wide}  \cite{zhou2018deep} \cite{zhou2019deep}. With the popularity of e-commerce shopping, e-commerce platform such as Taobao, the world's leading e-commerce platforms,  are now enjoying a huge boom in traffic \cite{cardellini1999dynamic}, e.g. user requests at Taobao are increasing year by year. As a result, the system load is under great pressures \cite{zhou2018rocket}. Moreover, request fluctuation also gives critical challenge to online serving system. For example, the Taobao recommendation system always bears many spikes of requests during the period of Double 11 shopping festival. 

\begin{figure}[thb]
\includegraphics[width=\columnwidth] {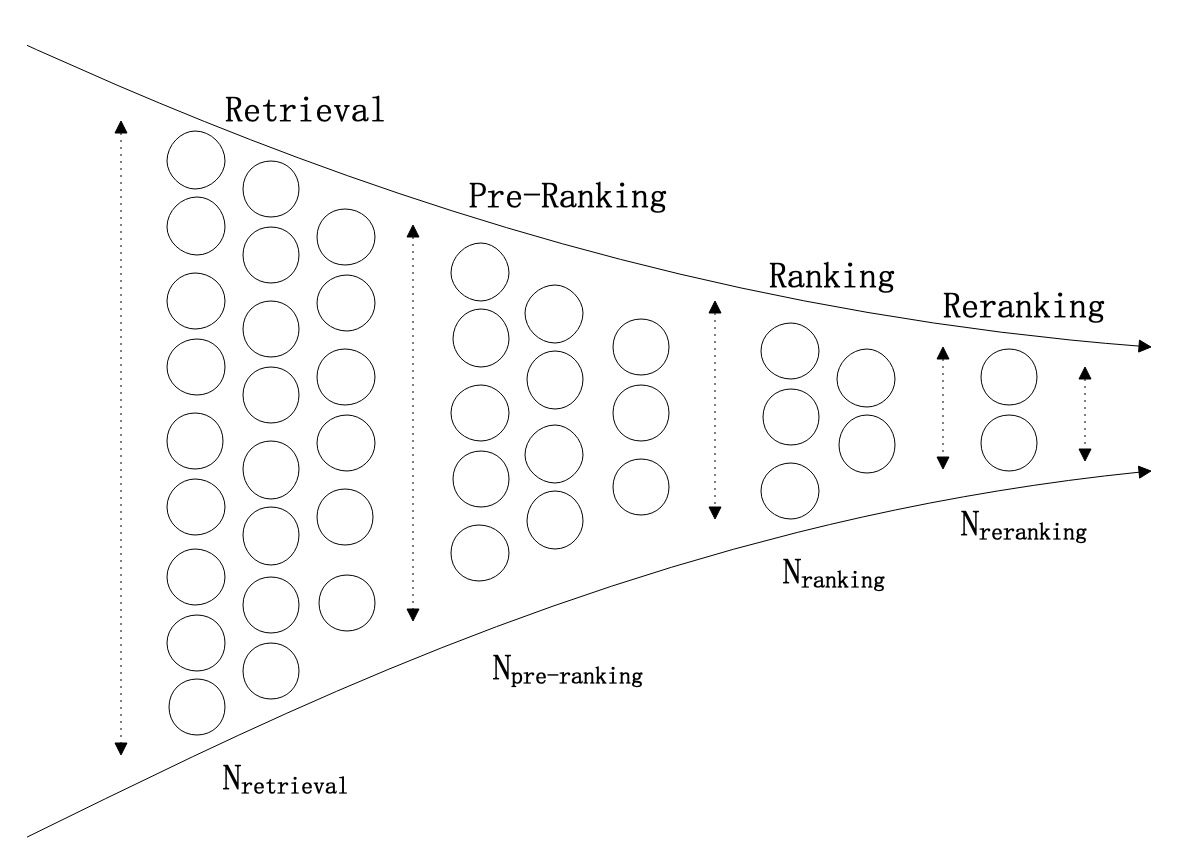}
\caption{\label{fig:cascade} Illustration of our cascaded display advertising system. Each request will be served through these modules sequentially. Considering the limitation of computation resource and latency for online serving system, the fixed quota of candidate advertisements, denoted by $N$ for each module, is usually pre-defined manually by experience. }
\medskip
\small
\end{figure}

To address the above challenges, the prevailing practices for online engine are: 1) decomposing the cascade system \cite{liu2017cascade} into multiple modules and manually allocating a fixed quota for each module by experience, as shown in Figure~\ref{fig:cascade}; 2) designing many computation downgrade plans in case the sudden traffic spikes arrive and manually executing these plans when needed. 

These non-automated strategies are often lack of flexibility and require human interventions. Furthermore, most of these practices often impact on all requests once they are executed and ignore the fact that the value of requests varies greatly. Obviously it is a straightforward but better strategy to allocate the computation resource by biasing towards the requests that are more valuable than others, for maximizing the total revenue. 

Considering the shortcomings of existing works, we aim at building a dynamic allocation framework that can allocate the computation budget flexibly among requests. Moreover, this framework should also take into account the stability of the online serving system which are frequently challenged by request boom and spike. Specifically, we formulate the problem as a knapsack problem, of which objective is to maximize the total revenue under a computation budget constraint. We propose a dynamic allocation framework named DCAF which could consider both computation budget allocation and stability of online serving system simultaneously and automatically. 

Our main contributions are summarized as follow:

\begin{itemize}
    \item We break through the stereotypes in most cascaded systems where each individual module is limited by a static computation budget independently. We introduce a brand-new idea that computation budget can be allocated w.r.t the value of traffic requests in a "personalized" manner. 
    \item  We propose a dynamic allocation framework DCAF which could guarantee in theory that the total revenue can be maximized under a computation budget constraint. Moreover, we provide an powerful control mechanism that could always keep the online serving system stable when encountering the sudden spike of requests.
    \item DCAF has been deployed in the display advertising system of Taobao, bringing a notable improvement. To be specific, system with DCAF maintains the same business performance with 20\% GPU (Graphics Processing Unit) resource reduction for online ranking system. Meanwhile, it greatly boosts the online engine's stability.
    \item By defining the new paradigm, DCAF lays the cornerstone for jointly optimizing the cascade system among different modules and raising the ceiling height for performance of online serving system further.
\end{itemize}

\section{Related work}
Quite a lot of research have been focusing on improving the serving performance. Park et al. \cite{park2018deep} describes practice of serving deep learning models in Facebook. Clipper \cite{crankshaw2017clipper} is a general-purpose low-latency prediction serving system. Both of them use latency, accuracy, and throughput as the optimization target of the system. They also mentioned techniques like batching, caching, hyper-parameters tuning, model selection, computation kernel optimization to improve overall serving performance. Also, many research and system use model compression \cite{han2015deep}, mix-precision inference, quantization \cite{gupta2015deep,courbariaux2015binaryconnect}, kernel fusion \cite{chen2018tvm}, model distillation \cite{hinton2015distilling, zhou2018rocket} to accelerate deep neural net inference.

Traditional work usually focus on improving the performance of individual blocks, and the overall serving performance across all possible queries. Some new systems have been designed to take query diversity into consideration and provide dynamic planning. DQBarge \cite{chow2016dqbarge} is a proactive system using monitoring data to make data quality tradeoffs. RobinHood \cite{berger2018robinhood} provides tail Latency aware caching to dynamically allocate cache resources. Zhang et al. \cite{zhang2019e2e} takes user heterogeneity into account to improve quality of experience on the web. Those systems provide inspiring insight into our design, but existing systems did not provide solutions for computation resource reduction and comprehensive study of personalized planning algorithms.

\section{Formulation}
We formulate the dynamic computation allocation problem as a knapsack problem which is aimed at maximizing the total revenue under the computation budget constraint. We assume that there are $N$ requests $\{i=1,\dots,N\}$ requesting the e-commerce platform within a time period. For each request, $M$ actions $\{j=1,\dots,M\}$ can be taken. We define  $Q_{ij}$ and $q_j$ as the expected gain for request $i$ that is assigned action $j$ and the cost for action $j$ respectively. $C$ represents the total computation budget constraint within a time period. For 
instance, in the display advertising system deployed in e-commerce, $q_j$ usually stands for items (ads) quota that request the online engine to evaluate, which positively correlate with system load in usual. And $Q_{ij}$ usually represent the eCPM~(effective cost per mille) conditioned on action $j$ which directly proportional to  $q_j$. $x_{ij}$ is the indicator that request $i$ is assigned action $j$. For each request, there is one and only one action $j$ can be taken, in other words, $x_{i.}$ is an one-hot vector. \newline
Following the definitions above, for each request, our target is to maximize the total revenue under computation budget by assigning each request $i$ with appropriate action $j$. Formally,
\begin{align*}
{\rm \max_j}&\sum_{ij}{x_{ij}Q_{ij}} \\
     &{\rm s.t.} \sum_{ij}{x_{ij}q_{j}} \le C \\
     &\sum_j{x_{ij}} \le 1 \\
     &x_{ij} \in \{0,1\} \numberthis
\end{align*}
where we assume that each individual request has its "personalized" value, thus should be treated differently. Besides, request expected gain is correlated with action $j$ which will be automatically taken by the platform in order to maximize the objective under the constraint. In this paper, we mainly focus on proving the effectiveness of DCAF's framework as a whole. However, in real case, we are faced with several challenges which are beyond the scope of this paper. We simply list them as below for considering in the future:
\begin{itemize}
    \item The dynamic allocation problem \cite{berger2018robinhood} are usually coupled with real-time request and system status. As the online traffic and system status are both varying with time, we should consider the knapsack problem to be real-time, e.g. real-time computation budget.
    \item $Q_{ij}$ are unknown, thus needs to be estimated. $Q_{ij}$ prediction is vital to maximize the objective, which means real-time and efficient approaches are required to estimate the value. Besides, to avoid increasing the system's burden, it is essential for us to consider light-weighted methods.
\end{itemize}

\section{Methodology}
\subsection{Global Optimal Solution and Proof}
To solve the problem, we firstly construct the Lagrangian from the formulation above,
\begin{align*}
L = -\sum_{ij}{x_{ij}Q_{ij}}+\lambda (\sum_{ij}{x_{ij}q_{j}}-C)+\sum_i{(\mu_i(\sum_j{x_{ij}}-1))} \\
= \sum_{ij}{x_{ij}(-Q_{ij}+\lambda q_j+\mu_i)}-\lambda C-\sum_i{\mu_i} \\
{\rm s.t.} \lambda \ge 0 \\
\mu_i \ge 0 \\
x_{ij} \ge 0 \numberthis
\end{align*}
where we relax the discrete constraint for the indicator $x_{ij}$, we could show that the relaxation does no harm to the optimal solution. From the primal, the dual function \cite{boyd2004convex} is
\begin{align}
    {\rm \max_{\lambda,\mu}}\ {{\rm \min_{x_{ij}}}}(\sum_{ij}{x_{ij}(-Q_{ij}+\lambda q_j+\mu_i)}-\lambda C-\sum_i{\mu_i}) 
\end{align}
With $x_{ij} \ge 0$ ($x_{ij} \le 1$ is implicitly described in the Lagrangian), the linear function is bounded below only when $-Q_{ij}+\lambda q_j+\mu_i\ge0$. And only when $-Q_{ij}+\lambda q_j+\mu_i=0$, the inequality $x_{ij} > 0$ could hold which means $x_{ij} = 1$ in our case (remember that the $x_{i.}$ is an one-hot vector). Formally,
\begin{align*}
{\rm \max_{\lambda,\mu}}(-\lambda C-\sum_i{\mu_i}) \\
{\rm s.t.} -Q_{ij}+\lambda q_j+\mu_i\ge0 \\
\lambda \ge 0 \\
\mu_i \ge 0 \\
x_{ij} \ge 0 \numberthis
\end{align*}
As the dual objective is negatively correlated with $\mu$, the global optimal solution for $\mu$ would be
\begin{align}\label{equ5}
\mu_{i} = {\rm \max_{j}}(Q_{ij}-\lambda q_j)
\end{align}
Hence, the global optimal solution to $x_{ij}$ that indicate which action $j$ could be assigned to request $i$ is
\begin{align}
j = {\rm arg\max_{j}}(Q_{ij}-\lambda q_j) 
\end{align}
From Slater’s theorem \cite{slater1950lagrange}, it can be easily shown that the Strong Duality holds in our case, which means that this solution is also the global optimal solution to the primal problem.

\subsection{Parameter Estimation}
\subsubsection{Lagrange Multiplier}~\\
The analytical form of Lagrange multiplier cannot be easily, or even possibly derived in our case. And meanwhile, the exact global optimal solution in arbitrary case is computationally prohibitive. However, under some general assumptions, simple bisection search could guarantee that the global optimal $\lambda$ could be obtained. Without loss of generality, we reset the indices of action space by following the ascending order of $q_j$'s magnitude.
\begin{assumption}\label{ass1}
$Q_{ij}$ is monotonically increasing with $j$.
\end{assumption}

\begin{assumption}\label{ass2}
$\nicefrac{Q_{ij}}{q_j}$ is monotonically decreasing with $j$.
\end{assumption}

\begin{lemma}\label{lem0}
Suppose Assumptions (\ref{ass1}) and (\ref{ass2}) hold, for each $i$, $\nicefrac{Q_{i{j_1}}}{q_{j_1}} \ge \nicefrac{Q_{i{j_2}}}{q_{j_2}}$ will hold if $\lambda_1 \ge \lambda_2$, where $j_1$ and $j_2$ are the actions that maximize the objective under $\lambda_1$ and $\lambda_2$ respectively.
\end{lemma}
\begin{proof}
As Equation (\ref{equ5}) and $\mu_{i}\ge0$, the inequality $Q_{ij}-\lambda q_j \ge 0$ holds. Equally, $\nicefrac{Q_{ij}}{q_j} \ge \lambda$ holds. Suppose $\nicefrac{Q_{i{j_1}}}{q_{j_1}} < \nicefrac{Q_{i{j_2}}}{q_{j_2}}$, we have $\nicefrac{Q_{i{j_2}}}{q_{j_2}} > \nicefrac{Q_{i{j_1}}}{q_{j_1}} \ge \lambda_1 \ge \lambda_2$. However, we could always find $j_2^*$ such that $Q_{i{j_2^*}} \ge Q_{i{j_2}}$ and  ${q_{j_2^*}} > {q_{j_2}}$ where $\nicefrac{Q_{i{j_1}}}{q_{j_1}} \ge \nicefrac{Q_{i{j_2^*}}}{q_{j_2^*}} \ge \lambda_1 \ge \lambda_2$ such that $Q_{i{j_2^*}} \ge Q_{i{j_2}}$  by following the Assumptions (\ref{ass1}) and (\ref{ass2}). In order words, $j_2$ is not the action that maximize the objective. Therefore, we have $\nicefrac{Q_{i{j_1}}}{q_{j_1}} \ge \nicefrac{Q_{i{j_2}}}{q_{j_2}}$. 
\end{proof}

\begin{lemma}\label{lem1}
Suppose Assumptions (\ref{ass1}) and (\ref{ass2}) could be satisfied, both $max\sum_{ij}{x_{ij}Q_{ij}}$ and its corresponding $\sum_{ij}{x_{ij}q_{j}}$ are monotonically decreasing with $\lambda$.
\end{lemma}
\begin{proof}
With $\lambda$ increasing, $\nicefrac{Q_{ij}}{q_j}$ is also increasing monotonically by Lemma (\ref{lem0}). Moreover, by Assumptions (\ref{ass1}) and (\ref{ass2}), we conclude that both $max\sum_{ij}{x_{ij}Q_{ij}}$  and its corresponding $\sum_{ij}{x_{ij}q_{j}}$ are monotonically decreasing with $\lambda$.
\end{proof}

\begin{theorem}
Suppose Lemma (\ref{lem1}) holds, the global optimal Lagrange Multiplier $\lambda$ could be obtained by finding a solution that make $\sum_{ij}{x_{ij}q_{j}} = C$ hold through bisection search.
\end{theorem}
\begin{proof}
By Lemma (\ref{lem1}), this proof is almost trivial. We denote the Lagrange Multiplier that makes $\sum_{ij}{x_{ij}q_{j}} = C$ hold as $\lambda^*$. Obviously, the increase of $\lambda^*$ will result in computation overload and the decrease of $\lambda^*$ will inevitably reduce max$\sum_{ij}{x_{ij}Q_{ij}}$ due to the monotonicity in Lemma (\ref{lem1}). Hence, $\lambda^*$ is the global optimal solution to the constrained maximization problem. Besides, the bisection search must work in this case which is also guaranteed by the monotonicity.
\end{proof}
Assumption (\ref{ass1}) usually holds because the gain is directly proportional to the cost in general,e.g. more sophisticated models usually bring better online performance. For Assumption (\ref{ass2}), it follows the law of diminishing marginal utility \cite{scott1955fishery}, which is an economic phenomenon and reasonable in our constrained dynamic allocation case.

The algorithm for searching Lagrange Multiplier $\lambda$ is described in Algorithm \ref{algo:lagrange}. In general, we implement the bisection search over a pre-defined interval to find out the global optimal solution for $\lambda$. Suppose $\min_j\sum_{j}{q_{j}} \le C \le \max_j \sum_{j}{q_{j}}$ (o.w there is no need for dynamic allocation), it can be easily shown that $\lambda$ locates in the interval $[0, \min_{ij}(\nicefrac{Q_{ij}}{q_{j}})]$. Then we get the global optimal $\lambda$ through bisection search of which target is the solution of $\sum_{ij}{x_{ij}q_{j}} = C$. 
\begin{algorithm}
\caption{\label{algo:lagrange} Calculate Lagrange Multiplier}
\begin{flushleft} 
1: \textbf{Input:}  $Q_{ij}$, $q_j$, $C$, interval $[0, \min_{ij}(\frac{Q_{ij}}{q_{j}})]$ and tolerance $\epsilon$ \\
2: \textbf{Output:} Global optimal solution of Lagrange Multiplier $\lambda$\\
3: Set $\lambda_l = 0$, $\lambda_r = \min_{ij}(\frac{Q_{ij}}{q_{j}})$, $gap = +\infty$ \\
4: \textbf{while} ($gap > \epsilon$): \\
5: \hspace{5mm} $\lambda_m = \lambda_l + \frac{\lambda_r - \lambda_l}{2}$ \\
6: \hspace{5mm} Choose action $j_m^*$ by 
\centerline{$\{j: {\rm arg\max_{j}}(Q_{ij}-\lambda_m q_j), Q_{ij}-\lambda_m q_j \ge 0\}$}\\
7: \hspace{5mm} Calculate the $\sum_i q_{j_l^*}$ denoted by $C_m$ \\
8: \hspace{5mm} $gap = |C_m - C|$ \\
9:  \hspace{5mm} \textbf{if} $gap \le \epsilon$: \\
10: \hspace{1cm} \textbf{return} $\lambda_m$ \\ 
11: \hspace{5mm} \textbf{else if}  $C_m \le C$: \\
12 : \hspace{1cm}  $\lambda_l = \lambda_m$ \\
13 : \hspace{5mm} \textbf{else}: \\
14:   \hspace{1cm}  $\lambda_r = \lambda_m$ \\
15: \textbf{end while} \\
16: \textbf{Return} the global optimal $\lambda_m$ which satisfies $|\sum_i q_{j_l^*} - C| \le \epsilon$.
\end{flushleft} 
\end{algorithm}

For more general cases, more sophisticated method other than bisection search, e.g. reinforcement learning, will be conducted to explore the solution space and find out the global optimal $\lambda$.

\subsubsection{Request Expected Gain}~\\
In e-commerce, the expected gain is usually defined as online performance metric e.g. Effective Cost Per Mile (eCPM), which could directly indicate each individual request value with regard to the platform. 
 Four categories of feature are mainly used: User Profile, User Behavior, Context and System status. It is worth noticing that our features are quite different from typical CTR model:
\begin{itemize}
    \item Specific target ad feature isn't provided because we estimate the CTR conditioned on actions.
    \item System status is included where we intend to establish the connection between system and actions.
    \item The context feature consists of the inference results from previous modules in order to re-utilize the request information efficiently. 
\end{itemize}

\section{Architecture}
\begin{figure}[htb]
\includegraphics[width=\columnwidth] {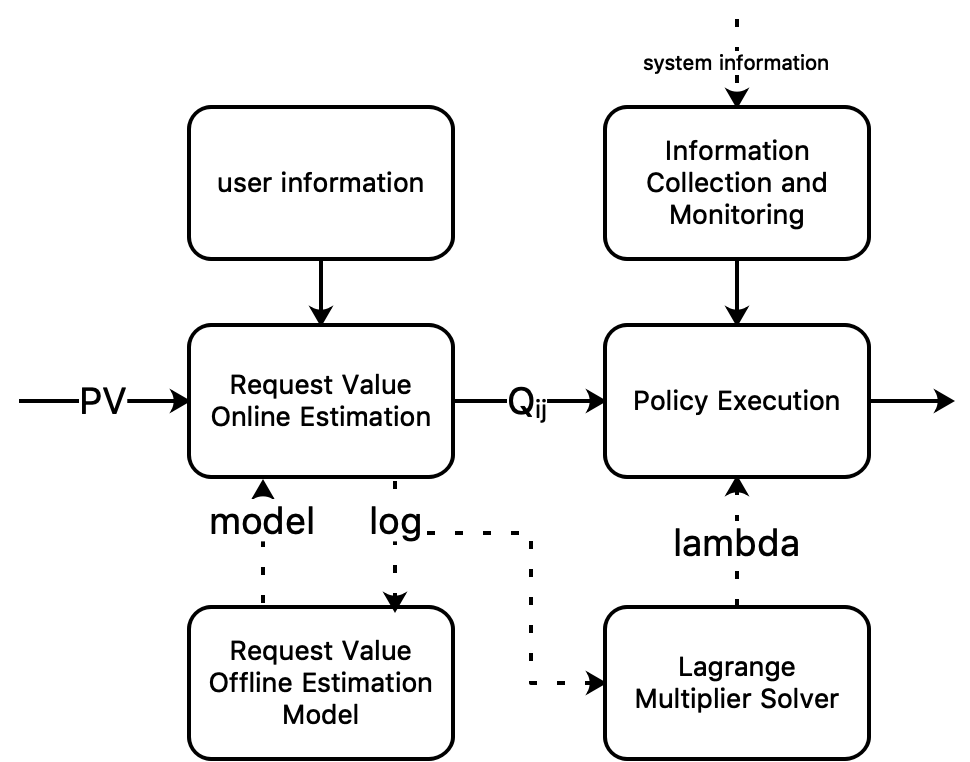}
\caption{\label{fig:framework} Illusion of the system of DCAF. \textit{Request Value Online Estimation} module will score each request conditioned on action $j$ through online features of which estimator is trained offline. \textit{Policy Execution} module mainly takes charge of executing the final action $j$ for each request based on the system status collected by \textit{Information Collection and Monitoring} module, $\lambda$ calculated offline and $Q_{ij}$ obtained from previous module.}
\medskip
\end{figure}
In general, DCAF is comprised of online decision maker and offline estimator:
\begin{itemize}
    \item The online modules make the final decision based on personalized request value and system status.
    \item The offline modules leverage the logs to calculate the Lagrange Multiplier $\lambda$ and train a estimator for the request expected value conditioned on actions based on historical data.
\end{itemize}
\subsection{Online Decision Maker}
\subsubsection{Information Collection and Monitoring}~\\
This module monitors and provides timely information about the system current status which includes GPU-utils, CPU-utils, runtime (RT), failure rate, and etc. The acquired information enables the framework to dynamically allocate the computation resource without exceeding the budget by limiting the action space. 
\subsubsection{Request Value Estimation}~\\
This module estimates the request's $Q_{ij}$ based on the features provided in information collection module. Notably, to avoid growing the system load, the online estimator need to be light-weighted, which necessitates the balance between efficiency and accuracy. One possible solution is that the estimation of $Q_{ij}$ should re-utilize the request context features adequately, e.g. high-level features generated by other models in different modules. 
\subsubsection{Policy Execution}~\\
Basically, this module assigns the best action $j$ to request $i$ by Equation (6). Moreover, for the stability of online system, we put forward a concept called \textit{MaxPower} which is an upper bound for $q_j$  to which each request must subject. DCAF sets a limit on the \textit{MaxPower} in order to strongly control the online engine. The \textit{MaxPower} is automatically controlled by system's runtime and failure rate through control loop feedback mechanism, e.g. Proportional Integral Derivative (PID) \cite{ang2005pid}. The introduction of \textit{MaxPower} guarantees that the system can adjust itself and remain stable automatically and timely when encountering sudden request spikes. \newline
According to the formulation of PID, $u(t)$ and $e(t)$ are the control action and system unstablity at time step $t$. For $e(t)$, we define it as the weighted sum of average runtime and fail rates over a time interval which are denoted by $rt$ and $fr$ respectively. $k_p$, $k_i$ and $k_d$ are the corresponding weights for  proportional,integral and derivative control. ${\theta}$ means a tuned scale factor for the weighted sum of $rt$ and $fr$. Formally,
\begin{align}
u(t)= k_p{e(t)}+k_i{\sum_{n=1}^{t}e(t)}+k_d({e(t)-e(t-1)})
\end{align}

\begin{algorithm}
\caption{PID Control for MaxPower}
\begin{flushleft} 
1: \textbf{Input:} $k_p$, $k_i$, $k_{d}$, $MaxPower$ \\
2: \textbf{Output:} $MaxPower$ \\
3: \textbf{while} (true): \\
4: \hspace{5mm} Obtain $rt$ and $fr$ from \textit{Information Collection and Monitoring} \\
5: \hspace{5mm} ${e(t)={rt}+\theta{fr}}$ \\
6: \hspace{5mm} $u(t)= k_p{e(t)}+k_i{\sum_{n=1}^{t}e(t)}+k_d({e(t)-e(t-1)})$\\
7: \hspace{5mm} Update $MaxPower$ with $u(t)$\\
8: \textbf{end while} \\
\end{flushleft} 
\end{algorithm}

\subsection{Offline Estimator}
\subsubsection{Lagrange Multiplier Solver} ~\\
As mentioned above, we could get the global optimal solution of the Lagrange Multipliers by a simple bisection search method. In real case, we take logs as a request pool to search a best candidate Lagrange Multiplier $\lambda$. Formally,
\begin{itemize}
    \item Sample $N$ records from the logs with $Q_{ij}$, $q_j$ and computation cost $C$, e.g. the total amount of advertisements that request the CTR model within a time interval.
    \item Adjust the computation cost $C$ by the current system status in order to keep the dynamic allocation problem under constraint in time. For example, we denote regular QPS by $QPS_r$ and current QPS by $QPS_c$. Then the adjusted computation cost $\hat{C}$ is $C \times \nicefrac{QPS_r}{QPS_c}$, which could keep the $N$ records under the current computation constraint.
    \item Search the best candidate Lagrange Multiplier $\lambda$ by Algorithm (1)
\end{itemize}
It's worth noting that we actually assume the distribution of the request pool is the same as online requests, which could probably introduce the bias for estimating Lagrange Multiplier. However, in practice, we could partly remove the bias by updating the $\lambda$ frequently.
\subsubsection{Expected Gain Estimator} ~\\
In our settings, for each request,  $Q_{ij}$ is associated with eCPM under different action $j$ which is the common choice for performance metric in the field of online display advertising. Further, we build a CTR model to estimate the CTR, because the eCPM could be decomposed into $ctr \times bid$ where the bids are usually provided by advertisers directly. It is notable that the CTR model is conditioned on actions in our case, where it is essential to evaluate each request gain under different actions. 
And this estimator is updated routinely and provides real-time inference in \textit{Policy Execution} module.

\section{experiments}
\subsection{Offline Experiments}
For validating the framework's correctness and effectiveness, we extensively analyse the real-world logs collected from the display advertising system of Taobao and conduct offline experiments on it. As mentioned above, it is common practice for most systems to ignore the differences in value of requests and execute same procedure on each request. Therefore, we set equally sharing the computation budget among different requests as the \textbf{baseline} strategy. 
As shown in Figure \ref{fig:cascade}, we simulate the performance of DCAF in Ranking stage by offline logs. In advance, we make it clear that all data has been rescaled to avoid breaches of sensitive commercial data. We conduct our offline and online experiments in Taobao's display advertisement system where we spend the GPU resource automatically through DCAF. In detail, we instantiate the dynamic allocation problem as follow:
\begin{itemize}
    \item Action $j$ controls the number of advertisements that need to be evaluated by online CTR model in Ranking stage.
    \item $q_j$ represents the advertisement's quota for requesting the online CTR model.
    \item $Q_{ij}$ is the sum of top-k ad's eCPM for request $i$ conditioned on action $j$ in Ranking stage which is equivalent to online performance closely. And $Q_{ij}$ is estimated in the experiment.
    \item $C$ stands for the total number of advertisements that are requesting online CTR model in a period of time within the serving capacity.
    \item \textbf{Baseline}: The original system, which allocates the same computation resource to different requests. With the baseline strategy, system scores the same number of advertisements in Ranking stage for each request.
\end{itemize}

\begin{figure}[htb]
\includegraphics[width=\columnwidth] {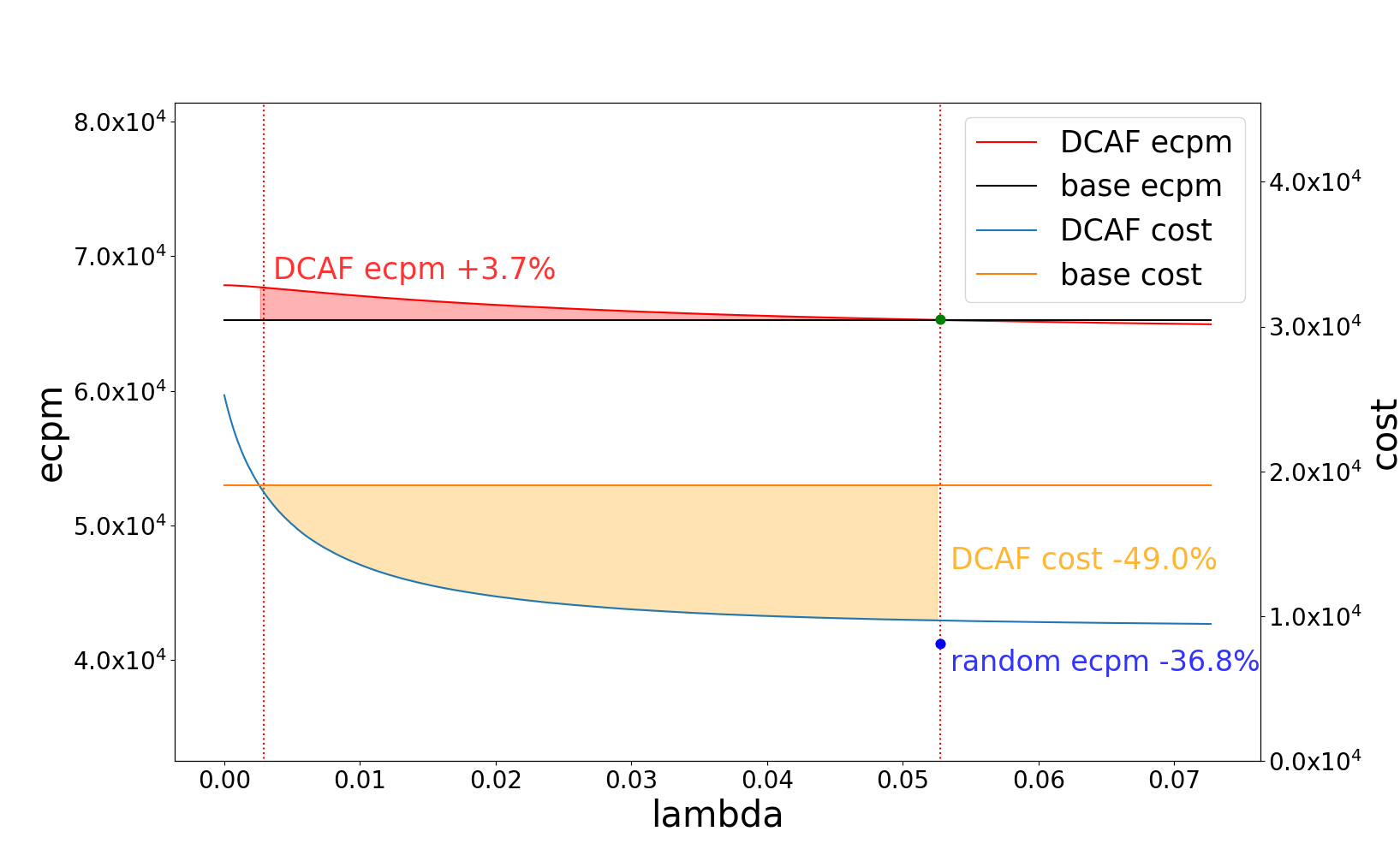}
\caption{\label{fig:exp1} Global optima under different $\lambda$ candidates. In Figure \ref{fig:exp1}, x-axis stands for $\lambda$'s candidate; left y-axis represents $\sum_{ij}{x_{ij}Q_{ij}}$; right-axis denotes the corresponding cost. The red shadow area corresponds to the exceeding part of $max\sum_{ij}{x_{ij}Q_{ij}}$ beyond the baseline. And yellow shadow area is the reduction of $\sum_{ij}{x_{ij}q_{j}}$ under these $\lambda$'s compared with the baseline. Random strategy is also shown in Figure \ref{fig:exp1} for comparison with DCAF.}
\medskip
\end{figure}

\textbf{Global optima under different $\lambda$ candidates. } In DCAF, the Lagrange Multiplier $\lambda$ works by imposing constraint on the computation budget. Figure \ref{fig:exp1} shows the relation among $\lambda$'s magnitude, $max\sum_{ij}{x_{ij}Q_{ij}}$ and its corresponding $\sum_{ij}{x_{ij}q_{j}}$ under fixed budget constraint. Clearly, $\lambda$ could monotonically impact on both  $max\sum_{ij}{x_{ij}Q_{ij}}$ and its corresponding $\sum_{ij}{x_{ij}q_{j}}$. And it shows that the DCAF outperforms the baseline when $\lambda$ locates in an appropriate interval. As demonstrated by the two dotted lines, in comparison with the baseline, DCAF can achieve both higher performance with same computation budget and same performance with much less computation budget. Compared with random strategy, DCAF's performance outmatches the random strategy's to a large extent.   \newline

\begin{figure}[htb]
\includegraphics[width=\columnwidth] {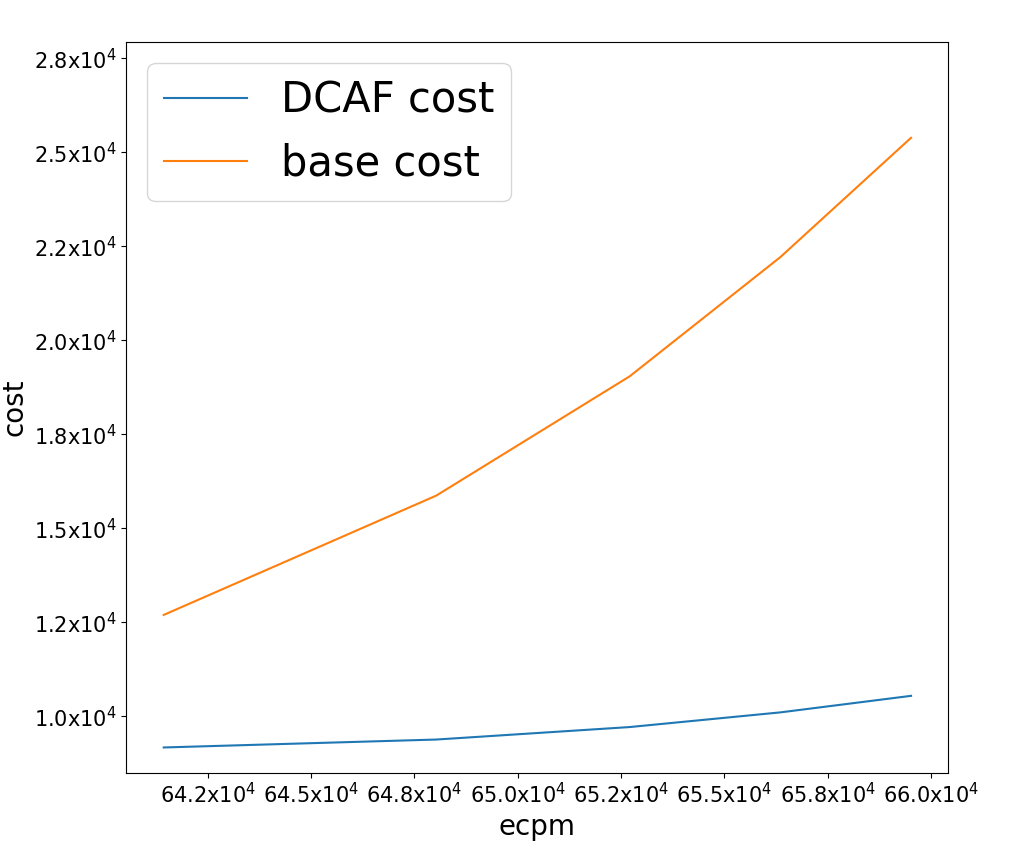}
\caption{\label{fig:exp2} Comparison of DCAF with the original system on computation cost. In Figure \ref{fig:exp2}, x-axis denotes the $\sum_{ij}{x_{ij}Q_{ij}}$; y-axis represents the$\sum_{ij}{x_{ij}q_{j}}$. For points on the two lines with same x-coordinate, Figure \ref{fig:exp2} shows that DCAF always perform as well as the baseline by much less computation resource.}
\medskip
\end{figure}

\textbf{Comparison of DCAF with the original system on computation cost. } Figure \ref{fig:exp2} shows that DCAF consistently accomplish same performance as the baseline and save the cost by a huge margin. Furthermore, DCAF plays much more important role in more resource-constrained systems.

\begin{figure}[htb]
\includegraphics[width=\columnwidth] {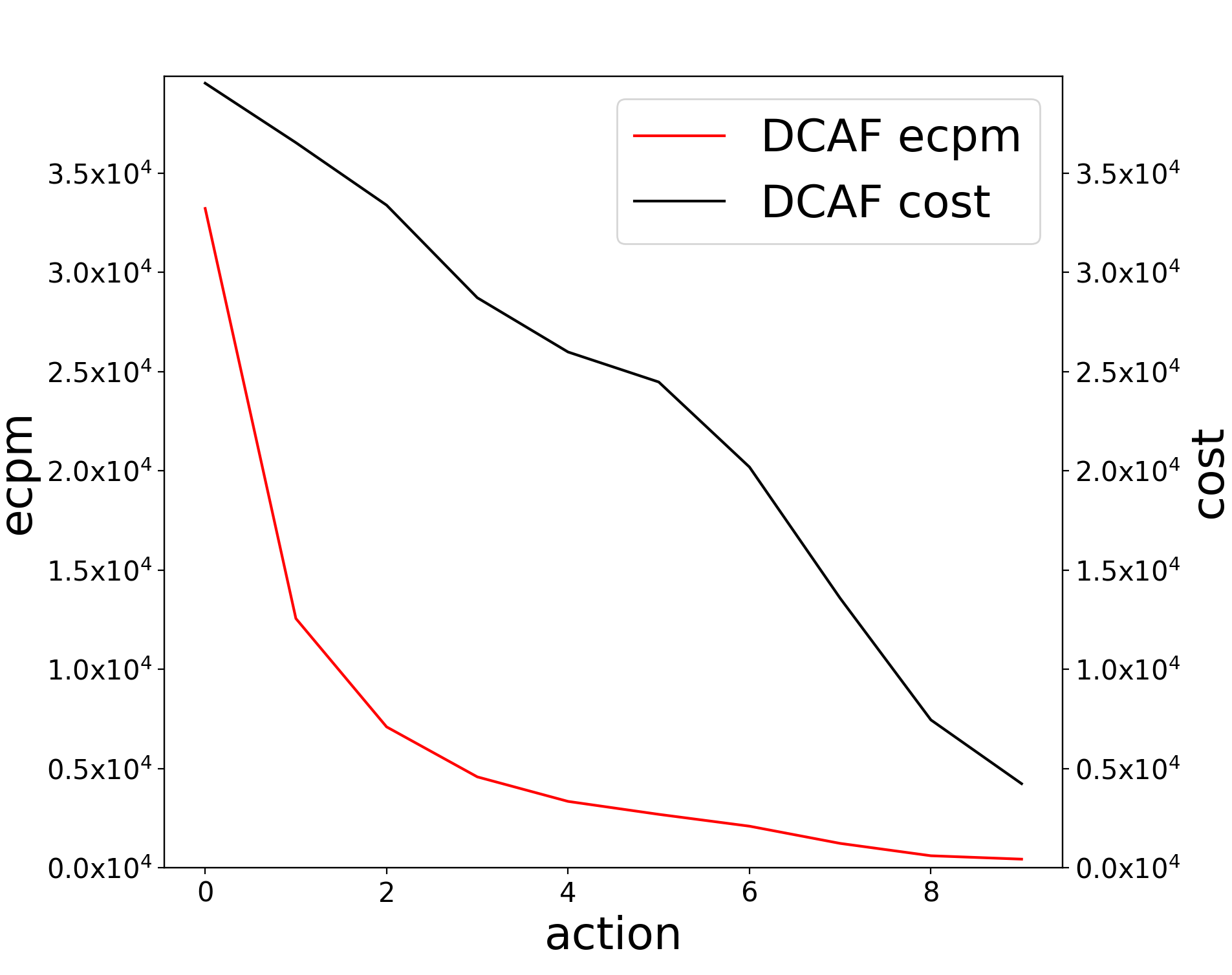}
\caption{\label{fig:exp3} Total eCPM and its cost over different actions. In this figure, x-axis stands for action $j$'s and left y-axis represents $\sum_{ij}{x_{ij}Q_{ij}}$ conditioned on action $j$; right-axis denotes the corresponding cost. For each action $j$, we sum over $Q_{ij}$ which is the the sum of top-k ad's eCPM for requests that are assigned action $j$ by DCAF.}
\medskip
\end{figure}

\textbf{Total eCPM and its cost over different actions. } As shown by the distributions in Figure \ref{fig:exp3}, we could see that DCAF treats each request differently by taking different action $j$. And $\nicefrac{\sum_{ij}Q_{ij}}{\sum_{ij}q_{j}}$ is decreasing with action $j$'s which empirically show that the relation between expected gain and its corresponding cost follows the law of diminishing marginal utility in total.

\subsection{Online Experiments}
DCAF is deployed in Alibaba display advertising system since 2020. From 2020-05-20 to 2020-05-30, we conduct online A/B testing experiment to validate the effectiveness of DCAF. The settings of online experiments are almost identical to offline experiments. Action $j$ controls the number of advertisements for requesting the CTR model in Ranking stage. And we use a simple linear model to estimate the $Q_{ij}$. The original system without DCAF is set as baseline. The DCAF is deployed between Pre-Ranking stage and Ranking stage which is aimed at dynamically allocating the GPU resource consumed by Ranking's CTR model. Table \ref{exp:samecost} shows that DCAF could bring improvement while using the same computation cost. Considering the massive daily traffic of Taobao, we deploy DCAF to reduce the computation cost while not hurting the revenue of the ads system. The results are illustrated in Table \ref{exp:samerpm}, and DCAF reduces the computation cost with respect to the total amount of advertisements requesting CTR model by 25\% and total utilities of GPU resource by 20\%. It should be noticed that, in online system, the $Q_{ij}$ is estimated by a simple linear model which may be not sufficiently complex to fully capture data distribution. Thus the improvement of DCAF in online system is less than the results of offline experiments. This simple method enables us to demonstrate the effectiveness of the overall framework which is our main concern in this paper. In the future, we will dedicate more efforts in modeling $Q_{ij}$. Figure \ref{fig:exp4} shows the performance of DCAF under the pressures of online traffic in extreme case e.g. Double 11 shopping festival. By the control mechanism of \textit{MaxPower}, the online serving system can react to the sudden rising of traffic quickly, and make the system back to normal status by consistently keeping the fail rate and runtime at a low level. It is worth noticing that the control mechanism of \textit{MaxPower} is superior to human interventions in the scenario that the large traffic arrives suddenly and human interventions inevitably delay. 

\begin{table}[!h]
  \caption{Results with Same Computation Budget}
\label{exp:samecost}    
  \begin{tabular}{lcl}
    \toprule
      & CTR & RPM\\
    \midrule
 Baseline & +0.00\% & +0.00\% \\ 
	DCAF & +0.91\% & +0.42\% \\
  \bottomrule
\end{tabular}
\end{table}

\begin{table}[!h]
  \caption{Results with Same Revenue}
\label{exp:samerpm}    
  \begin{tabular}{lcccl}
    \toprule
      & CTR & RPM & Computation Cost & GPU-utils\\
    \midrule
 Baseline & +0.00\% & +0.00\% & -0.00\% & -0.00\%\\ 
	DCAF & -0.57\% & +0.24\% & -25\% & -20\%\\
  \bottomrule
\end{tabular}
\end{table}

\begin{figure}[htb]
\centering
\begin{subfigure}[b]{\columnwidth}
   \includegraphics[width=1\linewidth]{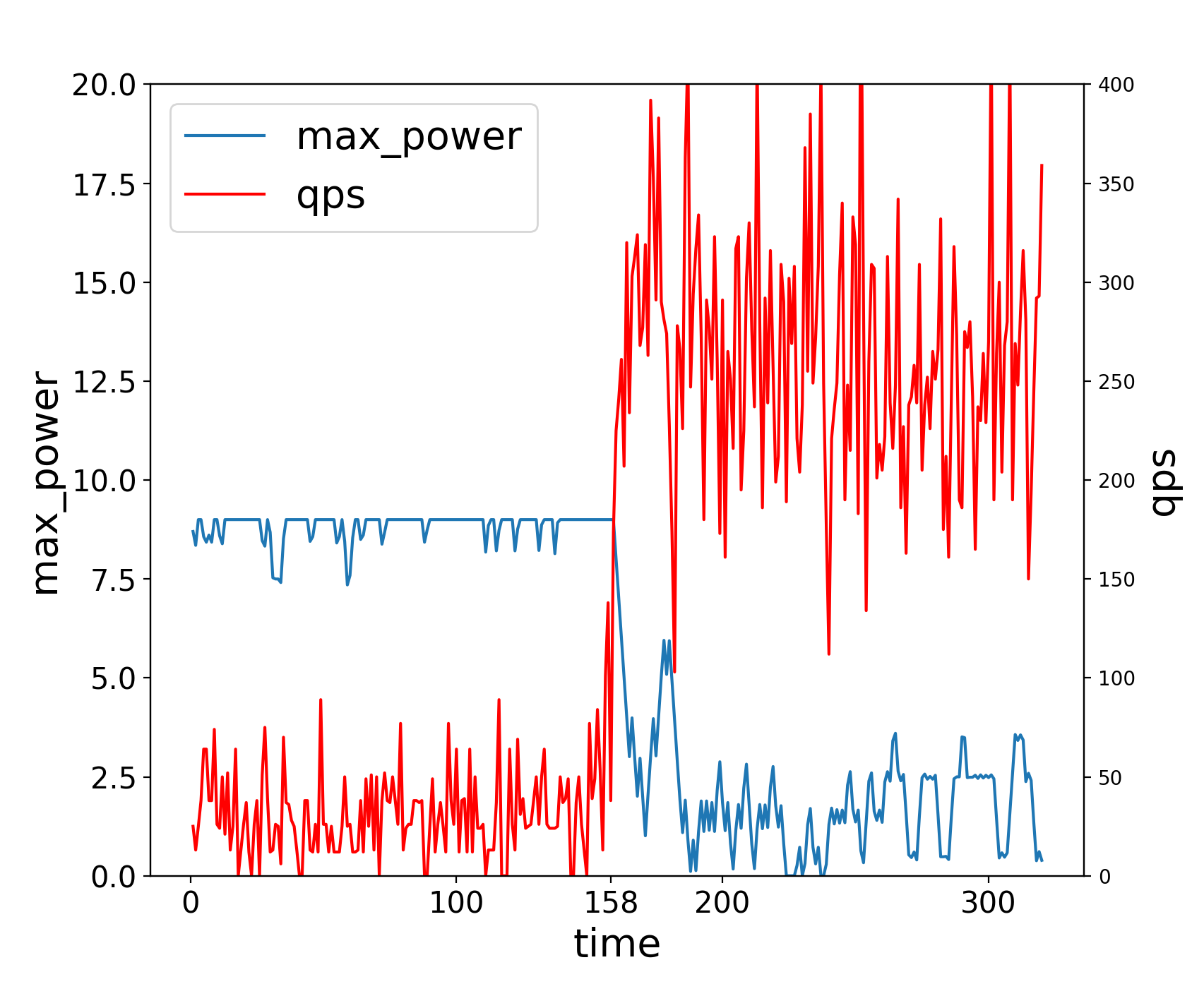}
   \caption{}
   \label{fig:Ng1} 
\end{subfigure}
\begin{subfigure}[b]{\columnwidth}
   \includegraphics[width=1\linewidth]{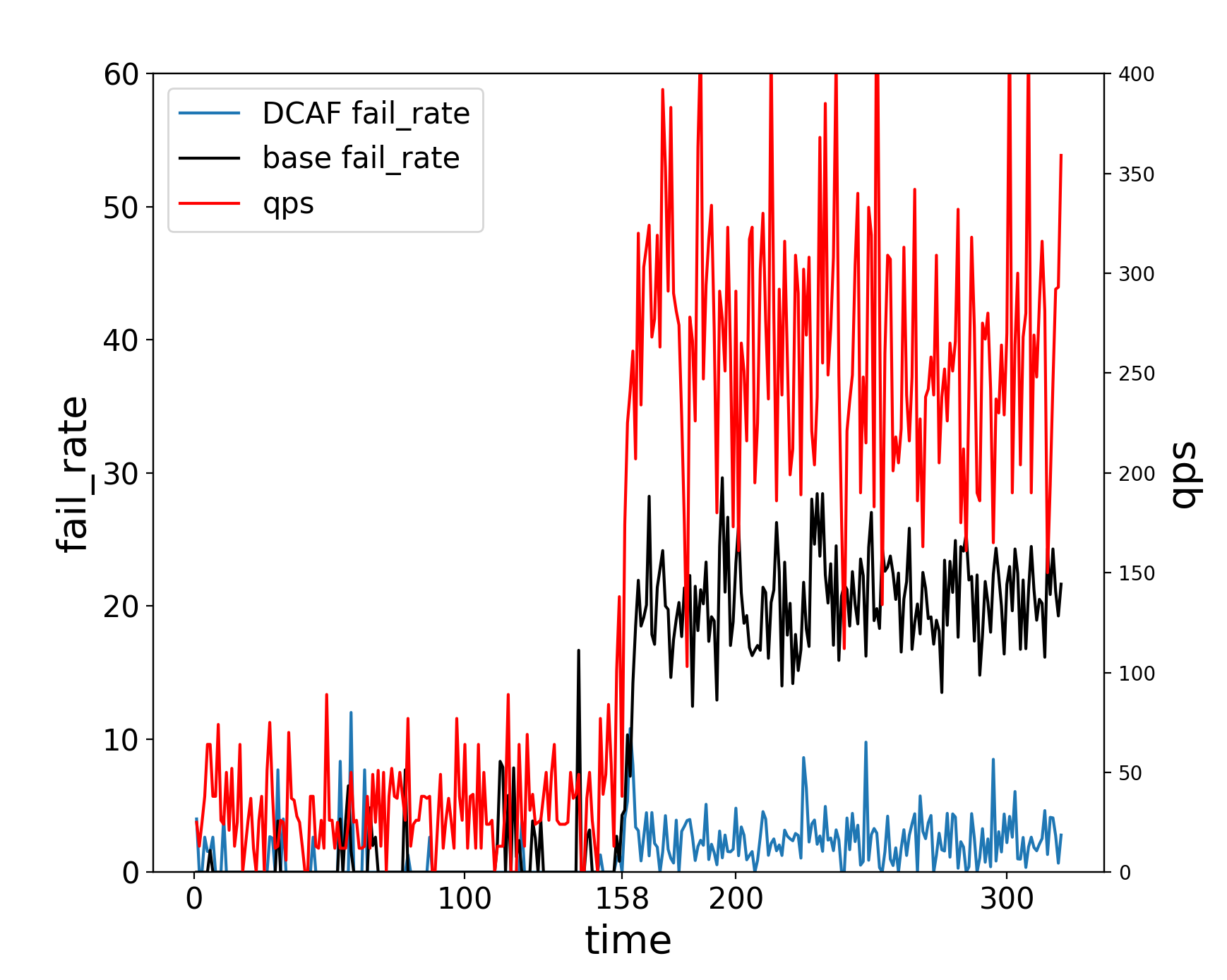}
   \caption{}
   \label{fig:Ng2}
\end{subfigure}
\caption{\label{fig:exp4} The effect of \textit{MaxPower} mechanism. In this experiments, we manually change the traffic of system at time $158$ and the requests per second increase 8-fold. Figure \ref{fig:Ng1} shows the trend of \textit{MaxPower} over time and Figure \ref{fig:Ng2} shows the trend of fail rate over time. As shown in Figure \ref{fig:exp4}, the \textit{MaxPower} takes effect immediately when the QPS is rising suddenly which makes the fail rate keep at a lower level. At the same time, the base strategy fails to serve some requests, because it does not change the computing strategy while the computation power of system is insufficient.}
\medskip
\end{figure}
\section{Conclusion}
In this paper, we propose a noval dynamic computation allocation framework (DCAF), which can break pre-defined quota constraints within different modules in existing cascade system. By deploying DCAF online, we empirically show that DCAF consistently maintains the same performance with great computation resource reduction in online advertising system, and meanwhile, keeps the system stable when facing the sudden spike of requests. Specifically, we formulate the dynamic computational allocation problem as a knapsack problem. Then we theoretically prove that the total revenue can be maximized under a computation budget constraint by properly allocating resource according to the value of individual request. Moreover, under some general assumptions, the global optimal Lagrange Multiplier $\lambda$ can also be obtained which finally completes the constrained optimization problem in theory. Moreover, we put forward a concept called \textit{MaxPower}  which is controlled by a designed control loop feedback mechanism in real-time. Through \textit{MaxPower} which imposes constraints on the range of action candidates, the system could be controlled powerfully and automatically.
\section{Future Work}
Fairness has attracted more and more concerns in the fields of recommendation system and online display advertisements. In this paper we propose DCAF, which allocate the computation resource dynamically among requests. The values of request vary with time, scenario, users and other factors, that incite us to treats each request differently and customize the computation resource for it. But we also noticed that DCAF may discriminate among users. While the allocated computation budgets varying with users, DCAF may leave a impression that it would aggravate the unfairness phenomenon of system further. In our opinion, the unfair problem stems from that all the approaches to model users are data-driven. Meanwhile 
most of systems create a data feedback loop that a system is trained and evaluated on the data impressed to users \cite{chaney2018algorithmic}. We think the fairness of recommender system and ads system is important and needs to be paid more attention to. In the future, we will analyse the long-term effect for fairness of DCAF extensively and include the consideration of it in DCAF carefully. \newline
Besides, DCAF is still in the early stage of development, where modules in the cascade system are considered independently and the action $j$ is defined as the number of candidate to be evaluated in our experiments. Obviously, DCAF could work with diverse actions, such as models with different calculation complexity. Meanwhile, instead of maximizing the total revenue in particular module, DCAF will achieve the global optima in the view of the whole cascade system in the future. Moreover, in the subsequent stages, we will endow DCAF with the abilities of quick adaption and fast reactions. These abilities will enable DCAF to exert its full effect in any scenario immediately.

\section{Acknowledgment}
We thanks Zhenzhong Shen, Chi Zhang for helping us on deep neural net inference optimization and conducting the dynamic resource allocation experiments.

\bibliographystyle{ACM-Reference-Format}
\bibliography{acmart}

\end{document}